\documentclass{llncs}

\usepackage{amssymb}
\usepackage{textcomp}
\usepackage{comment}
\usepackage{color}
\usepackage{graphicx}
\usepackage{amsmath}
\let\mathds\mathbb
\usepackage{url}
\usepackage{longtable}

\usepackage{algorithm}
\usepackage{algorithmic}

\newcommand{\E}{\mathbb{E}}

\newcommand{\cent}[0]{\mbox{\textcent}}

\DeclareMathOperator\Tr{Tr}

\begin{document}

\title{Deterministic Construction of QFAs based on the Quantum Fingerprinting Technique}
\author{Aliya~Khadieva$^{1,2}$ \and Mansur~Ziatdinov$^{1,3}$}

\institute{
Institute of Computational Mathematics and Information Technologies, Kazan Federal University, Kremlevskaya ul. 18, Kazan, Tatarstan,
420008 Russia
\and
University of Latvia, Riga, Latvia
\and
Zavoisky Physical-Technical Institute, FRC Kazan Scientific Center of RAS, Kazan, Russian
Federation
                \\ \email{aliyakhadi@gmail.com}
}
\maketitle

\begin{abstract}It is known that for some languages quantum finite automata are more efficient than classical counterparts. Particularly, a QFA recognizing the language $MOD_p$ has an exponential advantage over the classical finite automata. However, the construction of such QFA is probabilistic. In the current work, we propose a deterministic construction of the QFA for the language $MOD_p$. We construct a QFA for a promise problem $Palindrome_s$ and  implement this QFA on the IBMQ simulator using qiskit library tools.
\end{abstract}

\section{Introduction}

Quantum finite automaton is a mathematical model for quantum computers with limited memory. A quantum finite
automaton has a finite state space and applies a sequence of transformations corresponding to a letter of an input word
to its current  state. In the end, the state of the quantum automaton is measured, and the input word is accepted or rejected
depending on the outcome of the measurement.

Quantum automata are discussed and compared with classical ones in \cite{af98,ABGKMT06,AKV01,ABFK99,AY15,YS10A,sy2014,gy2018,AK03,ZQGLM13,AY12}.

In 1998, Freivalds and Ambainis presented a technique for construction of quantum finite automata exponentially smaller than classical counterparts for several languages \cite{af98}. In this work, they used a fingerprinting approach for construction of a QFA. This technique is also widely used for different models \cite{av2009,av2008,av2013,kk2017,l2009,l2006,kk2022,aaksv2022,kk2019disj,aavz2016,aakv2018}. 
In papers \cite{an2008,an2009}, Ambainis and Nahimovs used a similar technique and showed an improved exponential separation between quantum and classical finite automata. They constructed a QFA for the language $MOD_p$, define as $MOD_p = \{a^i|i$ is divisible by $p\}$, where $p$ is a prime.

In papers \cite{an2008,an2009}, a construction of the QFA  is probabilistic. It employs a sequence of
parameters  $(k_1 \cdots k_d)$ that are chosen at random and hardwired into the QFA. The idea of construction is the following. The QFA uses a register of $\log d+1$ qubits within one target qubit and $\log d$ control qubits.  A parameter $d$ depends on the probability of computational error.  At the beginning, the register is uniformly distributed on $d$ states. Reading input symbol the automaton makes a series of rotations of the target qubit $|q_{target}\rangle$ by angles $\alpha_i$ for each state $i \in \{0 \ldots d-1\}$. Each $\alpha_i$   depends on the parameter $k_i \in \{k_1 \cdots k_d\}$. 
Existence of such a set $\{k_1 \cdots k_d\}$ of parameters is proven in  \cite{an2008,an2009} by using  Azuma's theorem from Probability Theory.

In  papers  \cite{an2008,an2009}, authors gave two techniques of explicit construction of this set. The first of them gives a QFA with $O(\log p)$ states, but only numerical experiments show the correctness of the QFA. 
The second one gives an explicit construction of a QFA with $O(\log^{2+\epsilon}(p))$ states, where $\epsilon > 0$ is a probability of error. 

Vasiliev, Latypov and Ziatdinov \cite{lvz2017} used simulated annealing and genetic heuristic algorithms for construction of this set of required parameters $\{k_1 \cdots k_d\}$.
The algorithm for construction of a QFA was generalized by Ablayev and Vasiliev \cite{av2009}. They improved the technique and constructed an OBDD for a Boolean function $f=MOD_p(x_1,\dots, x_n)$, where $p$ is any positive  integer, not only prime. An OBDD is a modification of automata where transition functions depend on a position of an input head, and in OBDD, we can read input symbols in different orders \cite{Weg00}. Quantum and classical OBDDs were considered and compared in \cite{agk01,agkmp2005,ag05,a2006,fmp02,aakv2018,ikpy2018,ikpy2021,aakk2018,kk2017,agky14,agky16}.

In the current paper, we consider the same language $MOD_p$  and a QFA with $O(\log p)$ states recognizing this language.   It uses a set of parameters mentioned above. However, we propose a deterministic algorithm for computing the required parameters. It allows an explicit construction of a QFA for $MOD_p$ with $O(\log p)$ states. This technique is based on the results of \cite{wx2008}. 
We use a sequence of functions called pessimistic estimators. The theorem from \cite{wx2008} claims that if the sequence of such pessimistic estimators is defined, then the required set of parameters can be computed deterministically. Using this algorithm, we obtain explicit numbers as parameters for construction of the QFA. 

Additionally, we construct a QFA for the $Palindrome_s$ language as one more example for using this technique. Thus, we propose that this approach can be used for solving a large class of languages.

 The paper is organized in the following way. Section \ref{prelim} contains some preliminaries and previous results on an improved construction of a QFA. In Section \ref{widgerson} we give a method of deterministic finding of a set of parameters. 
 In Section \ref{palindrome}, there is one more example for applying the quantum fingerprinting technique. Here we describe a scheme of the QFA for the $Palindrome_s$ language. Futher, Section \ref{implementation} contains a short description of the QFA implementation on the IBMQ simulator and a link to the repo with the full code. Section \ref{conclusion}  is a conclusion.
%%%%%%%%%%%%%%%%%%%%%%%%%%%%%%%%%%%%%%%%%%%%
%              Preliminaries               %
%%%%%%%%%%%%%%%%%%%%%%%%%%%%%%%%%%%%%%%%%%%%
\section{Preliminaries}\label{prelim} 
We use a definition of a 1-way quantum finite automata (QFA) given in \cite{MC00}. A 1-way QFA is a tuple $M=(Q,\Sigma, \delta, q_0, Q_{acc} , Q_{rej})$ where $Q$ is a finite set of states, $\Sigma$ is an input alphabet, $\delta$ is a transition function, $q_0 \in Q$ is a starting
state, $Q_{acc}$ and $Q_{rej}$ are disjoint sets of accepting and rejecting states and $Q = Q_{acc} \cup Q_{rej}$. $\cent$ and $\mathdollar $ are symbols that do not belong to
$\Sigma$. We use  $\cent$ and $\mathdollar $  as the left and the right end markers, respectively. The working alphabet of $M$ is $\Gamma = \Sigma \cup \{\cent, \mathdollar\}$.

A superposition of $M$ is any element of $l_2(Q)$ (the space of mappings from $Q$ to $\mathds{C}^d$ with $l_2$ norm). For $q \in Q$, $|q\rangle$  denotes
the unit vector with value 1 at the position $q$ and 0
elsewhere. All elements $|\psi \rangle$ of $l_2(Q)$  can be expressed as linear combinations of vectors
$|q\rangle$. We will use $|\psi \rangle$ to denote elements of $l_2(Q)$.

A transition function $\delta$ maps $Q \times \Gamma \times Q$ to $\mathds{C}^d$. The value $\delta(q_1, a, q_2) $ is an amplitude of $|q_2\rangle$ in the superposition of
states to which $M$ goes from $|q_1\rangle$ after reading $a$. For $a \in \Gamma$ , $V_a$ is a linear transformation on $l_2(Q)$ defined by

\[V_a(|q_1\rangle) = \sum_{q_2 \in Q}{\delta(q_1, a, q_2)|q_2\rangle}\]

We require all $V_a$ to be unitary.
The computation of the QFA starts from the superposition $|q_0\rangle$. 
Then transformations corresponding to the left end marker $\cent$,
 the letters of the input word $x$ and the right end marker $\mathdollar$ are applied. 
 The transformation corresponding to $a \in \Gamma$ is just $V_a$.
If the superposition before reading $a$ is $|\psi\rangle$, then the superposition after reading $a$ is $V_a(|\psi\rangle)$.
After reading the right endmarker, the current state $|\psi\rangle$ is measured with respect to the observable $E_{acc} \oplus E_{rej} $ where
$E_{acc}  = span\{|q\rangle : q \in Q_{acc} \}$, $E_{rej} = span\{|q\rangle : q \in Q_{rej} \}$. This observation gives $|\psi \rangle \in E_i$ with the probability equal to the square
of the projection of $|\psi\rangle$ to $E_i$.

After that, the superposition collapses to this projection.
If we get $|\psi\rangle \in E_{acc}$, the input is accepted. If  $|\psi\rangle\in E_{rej}$, the input is rejected.

Let $L$ be a language in the alphabet $\Sigma$. We say a bounded error QFA recognizes a language $L$ iff after reading an input word $x$ the automaton terminates in an accepting state from $ Q_{acc}$.
\begin{itemize}
    \item  with probability more than $1-\epsilon$ if $x \in L$
    \item with probability less than $\epsilon$ if  $x \notin L$
\end{itemize}

In the work, we consider only one-sided error automata. In this case, the words $x \in L$ are accepted with probability 1 and words $x \notin L$ are accepted with probability at most $\epsilon$.

\subsection{Previous results}
Let $p$ be a prime number. We consider the language $MOD_p = \{a^i|i$ is divisible by $p\}$.
It is easy to see that any deterministic 1-way
finite automaton recognizing $MOD_p$ has at least $p$ states.
 Ambainis and Freivalds constructed an efficient QFA in \cite{af98}. They have shown that a QFA with $O(\log p)$ states can recognize $ MOD_p$ with bounded error $\epsilon$.
A big-O constant in this result depends on $\epsilon$. For $x \in MOD_p$, an answer is always correct with probability 1. There is a QFA with $16\log_2 p$ states that is correct with probability at least $1/8$ on inputs $x \notin MOD_p$.
In the general case, \cite{af98} gives a QFA with $poly(1/\epsilon)\log_2 p$ states that is correct with probability at least $1 -\epsilon$  on inputs $x \notin MOD_p$ (where $ poly(z)$ is some polynomial in $z$).

The papers \cite{an2008,an2009} present a simple design of QFAs that achieve a better big-O constant.
Ambainis and Nahimovs show that for any $\epsilon > 0$, there is a QFA with $4 \frac{\log_2 2p}{\epsilon}$
states recognizing $MOD_p$ with probability at least $1 - \epsilon$. Denote this QFA by $M$. It is constructed by combining $d$ subautomata $M_i$, each has 2 states $q_{i,0}$ and $q_{i,1}$. 

The set of states of $M$ consists of $2d$ states $\{q_{1,0}, q_{1,1}, q_{2,0}, q_{2,1},\cdots, q_{d,0}, q_{d,1}\}$. The starting state is $q_{1,0}$.
The set of accepting states $Q_{acc}$ consists of one state $q_{1,0}$. All other states $q_{i,j}$ belong to $Q_{rej}$.
A transformation for the left endmarker $\cent$ is such that
$V_{\cent}(|q_{1,0} \rangle)=\frac{1}{\sqrt{d}}(|q_{1,0} \rangle+|q_{2,0} \rangle+\cdots+|q_{d,0} \rangle).$

A transformation for the input symbol $a$ is
\begin{itemize}
    \item $V_a(|q_{i,0} \rangle) = \cos {\frac{2 k_i \pi }{p}}|q_{i,0} \rangle+\sin {\frac{2 k_i \pi }{p}}|q_{i,1} \rangle$
    \item $V_a(|q_{i,1} \rangle) = -\sin {\frac{2 k_i \pi }{p}}|q_{i,0} \rangle+\cos {\frac{2 k_i \pi }{p}}|q_{i,1} \rangle,$
\end{itemize} where $k_1,k_2,\cdots k_d$ is a sequence of numbers from $\{0,1,\cdot,p-1\}$.

A transformation for the endmarker $\mathdollar$ is the following.
\begin{itemize}
\item The states $ |q_{i,1} \rangle$  are not changed,

\item $V_{\mathdollar} (|q_{i,0} \rangle) = \frac{1}{\sqrt{d}} |q_{1,0} \rangle$ plus some other state.

In particular, 
$V_{\mathdollar}(\frac{1}{\sqrt{d}}(|q_{1,0} \rangle+|q_{2,0} \rangle+\cdots+|q_{d,0} \rangle))= |q_{1,0} \rangle.$
\end{itemize}
If the input word is $a^j$ and $j$ is divisible by $p$, then $M$ accepts the word with probability $1$.
If the input word is $a^j$ and $j$ is not divisible by $p$, then $M$ accepts with probability
\[\frac{1}{d^2} \left(\cos {\frac{2 k_1 \pi j}{p}}+\cos {\frac{2 k_2 \pi j }{p}}+\cdots+\cos {\frac{2 k_d \pi j}{p}}\right)^2.\]
In their proof, the authors use a theorem from Probability Theory (variant of Azuma's theorem):
\begin{theorem}\cite{mrr95}
 Let $X_1, \cdots , X_d$ be independent random variables such that $E[X_i] = 0$ and the value of $X_i$
is always between $-1$ and $1$. Then,

\[Pr\left[| \sum_{i=1}^{d} {X_i}| \geq \lambda\right] \leq 2e^{-\frac{\lambda^2}{2d}}.\]  
\end{theorem}

Define $X_i$ as $\cos \frac{2 k_i \pi j}{p}$ where each $k_i $ is from $\{0, \cdots p-1\}$. By the theorem, it is possible to choose $k_1, \cdots, k_d$ values to ensure

\begin{equation} \label{eq1}
    \frac{1}{d^2} \left(\sum_{i=1}^d{\cos {\frac{2 k_i \pi j }{p}}}\right)^2< \epsilon.
\end{equation}

This inequality gives a bound for $d = 2\frac{\log_2 2p}{\epsilon}$, and thus, a number of states for $M$ is $4\frac{\log_2 2p}{\epsilon}$.

The proposed QFA construction depends on $d$ parameters and accepts an input word $a^j \notin MOD_p$ with probability \[\frac{1}{d^2} \left(\sum_{i=1}^d{\cos {\frac{2 k_i \pi j }{p}}}\right)^2\]

However, this proof is by a probabilistic argument and does not give an explicit sequence $k_1, \dots, k_d$.

 Two approaches for construction of specific sequences are presented in \cite{an2008,an2009}.
 
 The first one is based on numerical experiments and gives a QFA with $O(\log p)$ states.
 It is based on using so-called cyclic sequences $S_g=\{k_i = g^i$ mod $p\}^d_{i=1}$ where $g$ is a primitive root modulo $p$. The authors checked all $p \in \{2, \dots , 9973\}$, all generators $g$ and all sequence lengths $d < p$. They experimentally compared two strategies: using a randomly chosen sequence and using a
cyclic sequence.   In $99.98\%-99.99\%$ of the experiments, random sequences satisfied the theoretical upper bound (\ref{eq1}), but cyclic sequences
substantially outperformed random ones in almost all the cases. For some $p$, in $1.81\%$ of the cases, the random sequences gave better values of $\epsilon$. The numerical experiments showed that almost all the observed sequences satisfied the bound (\ref{eq1}). 

 However, it is still open 
whether one could find the desired generator without an exhaustive search of all generators for all values of $p$.
 
 The second approach gives an explicit construction of a QFA. This approach is based on a result of Ajtai et al.  \cite{a1990}. Nevertheless, the QFA has $O(\log^{2+3\epsilon}p)$ states that is more than $O(\log p)$.

\section{Deterministic Algorithm for Construction of a QFA} \label{widgerson}

In this section, we suggest an explicit algorithm for deterministic construction of the set of parameters $S=(k_1,\dots,k_d)$.
By using this approach we can explicitly construct a QFA for $MOD_p$ with $O(\log p)$ states. The   QFA is designed as in \cite{an2008,an2009}. However, the set $S$ of required parameters is not randomly chosen, but deterministically computed.

The method is based on an algorithm from \cite{wx2008} which is the explicit algorithm for deterministic construction of a small-biased set.

\subsection{Definitions}

Let us introduce some necessary definitions.

\begin{definition}
  Let us denote \([n]\) the set of integers \(\{0,\ldots,n-1\}\).
\end{definition} 

\begin{definition}
  We let \(A \ge 0\) denote that \(A\) is positive semidefinite (p.s.d.) matrix (i.e. all
  of its eigenvalues are non-negative).

  Let \(A\) and \(B\) be symmetric matrices. Let \(A \le B\) denote that \(B -
  A \ge 0\).

  Let us denote \([A;B]\) the set of all symmetric matrices such that \(A \le
  C\) and \(C \le B\).
\end{definition}

\begin{definition}[Pessimistic estimators]\label{defn:pessimistic}
  Let $\sigma : [p]^d \to \{0,1\}$ be an event.
  
  Pessimistic estimators $\phi_0, \ldots, \phi_d$ are functions $\phi_i : [p]^i
  \to [0,1]$, such that for each $i$ and for each $x_1, \ldots, x_i \in [p]$:
  \begin{equation}\label{eq:probphi}
    \Pr_{X_{i+1},\ldots, X_d}[\sigma(x_1, \ldots, x_i, X_{i+1}, \ldots, X_d) = 0] \le \phi_i(x_1,\ldots,x_i)
  \end{equation}
  and
  \begin{equation}\label{eq:ephi}
    \E_{X_{i+1}} \bigg[ \phi_{i+1}(x_1,\ldots,x_i, X_{i+1}) \bigg] \le \phi_i(x_1, \ldots, x_i)
  \end{equation} 
\end{definition}

\subsection{Algorithm for Finding a Small-Biased Set}

To compute a set of parameters we generalize a problem. Let \(H\) be a group and
\(\chi : H \to \mathbb{C}\) be a characteristic of \(H\). Suppose that we are
able to find a set \(S \subset H\) such that
\[
  \frac 1 {|S|} \bigg| \sum_{s\in S} \chi(s) \bigg| \le \gamma.
\]

If we take \(H = \mathbb{Z}_p\) then \(\chi(h) = \cos \frac{2\pi h}{p} + i \sin \frac{2\pi h}{p}\).

Therefore,
\begin{align*}
  \gamma &\ge \frac 1 {|S|} \bigg| \sum_{s\in S} \chi(s) \bigg| \\
           &= \frac 1 d \bigg| \sum_{j=1}^d \cos \frac{2\pi k_j}{p} + i \sin \frac{2\pi k_j}{p} \bigg| \\
           &\ge \frac 1 d \bigg| \sum_{j=1}^d \cos \frac{2\pi k_j}{p} \bigg|,
\end{align*}
where \(S = \{k_1, \ldots, k_d\}\) is the desired set of parameters.

Given $\gamma < 1$, we want to find  a set $S \subset [p]^d$ such that
\[
  \frac{1}{|S|} \bigg| \sum_{s \in S} \chi(s) \bigg| \le \gamma,
\]
where $p$ is a non-negative number, $d$ is a size of the desired set of
parameters $S$, $\epsilon$ is a probability of error.

The solution for this problem is given by Wigderson \cite{wx2008} in the
following theorem and its proof. Let us reformulate this theorem in our
notation.

\begin{theorem}[\cite{wx2008}]\label{thm:estimators}
  Let $\sigma : [p]^d \to \{0,1\}$ be an event.

  If there are pessimistic estimators $\phi_0, \ldots, \phi_d$ of $\sigma$, then
  we can efficiently find $x_1,\ldots,x_d$ by deterministic algorithm such that
  $\sigma(x_1,\ldots,x_d)=1$.
\end{theorem} 

\begin{proof}
Pick $x_1, \ldots, x_d$ one by one.

\begin{itemize}
  \item At step $0$ $\phi_0 < 1$.

  \item At step $i$ we have $x_1, \ldots, x_i$ and $\phi_i(x_1,\ldots,x_i) < 1$.
  Enumerate over $x_{i+1} \in [p]$ and choose a value such that
  $\phi_{i+1}(x_1,\ldots,x_{i+1}) \le \phi_i(x_1,\ldots,x_i) < 1$. An existence
  of $x_{i+1}$ is guaranteed by inequality (\ref{eq:ephi}).

  \item At step $d$ we have $x_1, \ldots, x_d$ and $\phi_d(x_1, \ldots, x_d) < 1$. By
  inequality (\ref{eq:probphi}), $\Pr[\sigma(x_1,\ldots,x_d)=0] \le
  \phi_d(x_1,\ldots,x_d) < 1$, therefore $\sigma(x_1,\ldots,x_d) = 1$.
  \end{itemize}
\end{proof} 

The definition of required functions \(f\) can be found in \cite[Theorem
5.2]{wx2008}. Let us reformulate it in our notation.

\begin{theorem}[\cite{wx2008}]
  Let $\gamma < 1$. Let $H$ be a finite abelian group with character $\chi : H \to \mathbb{C}$.  Denote $p = |H|$.

  There exists a set $S$ of size $|S| = O(\frac{1}{\gamma} \log n)$ such that
  \[
    \frac{1}{|S|} \bigg| \sum_{s\in S} \chi(s) \bigg| \le \gamma,
  \]
  and this set can be found deterministically.
\end{theorem}

\begin{proof}
  Let $P_h$ for $h \in H$ be a $p\times p$ permutation matrix of the action of
  $h$ by right multiplication.
  Let matrix-valued $f : H \to [-I_p, I_p]$ be $f(h) = (I - J/p)\frac{1}{2}(P_h
  + P_{h^{-1}})$, where $I_p$
  is identity $p \times p$ matrix and $J$ is the all-one matrix.
  Let $\sigma : [p]^d \to \{0,1\}$ be the event $\sigma(x_1, \ldots, x_d) = 1$
  iff $\frac 1 d \sum_{i=1}^d f(x_i) \le \gamma I$. Let \(t = \gamma / 2\).

  By \cite[Theorem 4.1]{wx2008} the following functions are pessimistic estimators for \(\sigma(x_1,\ldots,x_d)\):
  \begin{eqnarray*}
    \phi_0 &=& pe^{-t\gamma d} \bigg|\bigg| \E [ \exp(tf(X)) ] \bigg|\bigg|^d
               \le pe^{-\gamma^2 d / 4} \\
    \phi_i(x_1,\ldots,x_i) &=& pe^{-t\gamma d} \Tr \bigg( \exp \big( t
                               \sum_{j=1}^i f(x_j) \big) \bigg) \bigg|\bigg| \E [ \exp(tf(X)) ] \bigg|\bigg|^{d-i}
  \end{eqnarray*} 
  Therefore, by Theorem
  \ref{thm:estimators}, there exists polynomial algorithm to find $x_1, \ldots,
  x_d$ such that $\sigma(x_1,\ldots, x_d) = 1$.
\end{proof}

So, we get the following deterministic Algorithm~ \ref{fig:algo} for computing a set of parameters. Here, $p$ is  prime, $d$ is a size of the desired set of parameters $S$, $\gamma$ is a probability of error, $Z$ is an  array of size $p$ consisting of $p \times p$ zero matrices. If we set $\gamma=\sqrt{\epsilon}$, then the condition \ref{eq1} is met.

\begin{algorithm}
  \caption{Algorithm to find a set of parameters}
  \label{fig:algo}
  \begin{algorithmic}
  \REQUIRE \(p, d, \gamma\)
  \ENSURE \(S = \{k_1, \ldots, k_d\}\), s.t. \(\frac 1 d \bigg| \sum_{i=1}^d \cos \frac{2\pi k_i}{p} \bigg| \le \gamma\)
  \STATE \(t \gets \gamma/2\)
  \STATE \(f \gets Z\)
  \FORALL{\(i \in \{1, \ldots, p\}\)}
  \STATE \(f[i] \gets (I - J/n) (P_i + P_{i^{-1}}) / 2\)
  \ENDFOR
  \STATE \(S \gets \{\}\)
  \FORALL{\(i \in \{1,\ldots,d\}\)}
  \STATE \(m \gets 0\)
  \FORALL{\(j \in \{1,\ldots,p-1\}\)}
  \IF { \(\phi_{i+1}(s_1, \ldots, s_i, j) < \phi_i(s_1, \ldots, s_i)\) } 
  \STATE \(m \gets j\)
  \ENDIF
  \ENDFOR
  \STATE \(S \gets S \cup \{m\}\)
  \ENDFOR
  \RETURN \(S\)
  \end{algorithmic}
\end{algorithm}

See Appendix \ref{experiment} for numerical results.

\section{One More Example for Applying  the Quantum Fingerprinting Technique}\label{palindrome}

The paper \cite{akv2008} gives a notation of  Equality-related problems in a context of quantum Ordered Binary Decision Diagrams. Authors apply a new modification of the fingerprinting technique to such problems as Equality,  Palindrome, Periodicity, Semi-Simon, Permutation Matrix Test Function. We conclude that the given deterministic method of finding a required set of parameters works for the mentioned class of problems. In this section, we present a construction of a QFA for a promise problem $Palindrome_s $.

Define the promise problem $Palindrome_s$ as follows. Let $s$ be some even integer, an input is a  string $X$ in the alphabet $\{0,1,\#\}$. The input word is bordered by a left end marker $\cent$ and a right end marker $\mathdollar$. We are promised that after each $s$ input symbols belonging to $\{0,1\}$ a symbol $\#$ appears.
We split the input into so called subwords of size $s$ divided by the symbol $\#$. 
\[x_1x_2\ldots x_{s}  \# x_{s+1} \ldots x_{2s}\# \ldots \# x_{is+1} x_{is+2} \ldots x_{(i+1)s}\# \ldots\]

Define these subwords as $\omega_0, \omega_1, \ldots, \omega_i, \ldots $. If a length of the input is not divisible by $s$, then the remaining part of the input is ignored. 

Define  $Palindrome(x_1, x_2,\ldots x_{s}) \equiv (x_1 x_2\ldots x_{s/2}=x_{s/2+1}x_{s/2+2} \ldots x_{s})$.
\[Palindrome_s \equiv [Palindrome (\omega_0) \& Palindrome (\omega_1) \&  \ldots ]
\]

For solving the promise problem $Palindrome_s$, we consider a QFA that allows  measuring several times. 
Construction of the automaton is based on the technique described above.  Let the automaton be denoted by $P$.
Here, we use a quantum register $|\psi \rangle$ of $t=\lceil \log_2 d \rceil+1$ qubits, where $d=2\frac{\log_2 2 \cdot 2^{s/2}}{\epsilon}$ is a size of the set of parameters $S$. An additional register $|\phi \rangle$ of $\lceil\log_2 s \rceil$ qubits is needed for storing index of an observing symbol in a subword.
The QFA $P$ consists of $2^t\cdot s \approx 2ds$ states.
\[|\psi_1 \psi_2\ldots\psi_{t-1}\rangle |\psi_{target}\rangle |\phi \rangle=|00\ldots0\rangle |0\rangle|0\rangle\]
is an initial state.
An accepting state is $|00\ldots0\rangle |0\rangle|0\rangle$. All other states are rejecting.

Reading the left end marker $P$ maps the initial state to a superposition of $d$ states by applying the  Hadamard transformation to the first $t-1$ qubits.
\[|00\ldots0\rangle |0\rangle|0\rangle \rightarrow{}\frac{1}{\sqrt{d}} \sum_{i=0} ^{d} |i\rangle|0\rangle|0\rangle\]

The register $|\phi\rangle$ is changed after reading a symbol $x_j$ of a subword by applying a transformation
$U: |j\rangle \to |(j+1)\mbox{ }mod\mbox{ } s\rangle$, where $j\in\{0,\dots, s-1\}$.  The $s \times s$ matrix $U$ for the transformation is the following. 
\[U=
\begin{bmatrix}
    0 & 0 & \cdots & 0 & 1\\
    1 & 0 & \cdots & 0 & 0\\
    0 & 1 & \cdots & 0 & 0\\
    && \cdots &&\\
    0 & 0 & \cdots & 1 & 0\\
\end{bmatrix}\]

For a symbol $x_j=1$, the qubit  $|\psi_{target}\rangle$ is transformed  by  rotating $G$ by an angle $\frac{2 \pi k_i 2^{j}}{2^{s/2}}$ if $j < s/2$ (mod $s$), and on by angle $-\frac{2 \pi k_i 2^{s-j-1}}{2^{s/2}}$ if $j \geq s/2$ (mod $s$).  The transformation is the following.
\[G=
\begin{bmatrix}
    \cos \alpha       & \sin\alpha \\
   -\sin\alpha      & \cos\alpha
\end{bmatrix},\]

where $\alpha=\frac{2 \pi k_i 2^{j}}{2^{s/2}}$ if $j < s/2$ (mod $s$),

$\alpha=-\frac{2 \pi k_i 2^{s-j-1}}{2^{s/2}}$ if $j \geq s/2$ (mod $s$).

For $x_j=0$, the system is not transformed.

The set  $S=\{k_1,k_2 \ldots k_d\}$ can be computed in a way described in the previous section.

Reading the symbol $\#$ the automaton maps the first $t-1$ qubits to the $|00\ldots0\rangle$ by applying the Hadamard transformation, and $|\psi_{target}\rangle$ is measured. If the measured value is $|0\rangle$, then the computation is continued on the next subword starting from the same initial state $|00\ldots0\rangle |0\rangle|0\rangle$. Otherwise,  $P$ stops and rejects the input.

\begin{theorem}
If the input word $x \in Palindrome_s $, then $P$ accepts it
with probability $1$.
If  $x \notin Palindrome_s $, then $P$ accepts it with  probability $\epsilon$. 
\end{theorem}
\begin{proof}
  If $x \in Palindrome_s$, then each of its subwords is a palindrome. Initially, $|\psi_{target}\rangle=|0\rangle$. By the construction of the QFA, after reading each $\omega_i$, the automaton $P$ maps $|\psi_{target}\rangle$ to itself.
  After reading the right end-marker, $P$ gets into the state $|00\ldots0\rangle |0\rangle|0\rangle$. That is the only accepting state.
  
  If $x \notin Palindrome_s$, then at least one subword of the input is a non-palindrome. Let it be $\omega_g=x_1x_2\ldots x_s$. After reading $\omega_g$ and measuring $|\psi_{target}\rangle$, we get $|0\rangle$ with probability

  \[\frac{1}{d^2}\left(\sum_{i=0}^d \cos \frac{2 \pi k_i (q-r)}{2^{s/2}}\right)^2<\epsilon, \quad \textrm{where} \quad q=\sum\limits_{j=0}^{s/2-1} 2^{j}x_j, r= \sum\limits_{j=s/2}^{s-1}2^{s-j-1} x_j. \]

  Thus, $P$ accepts the wrong word with the probability less than $\epsilon$. 
  
  If there are $m$ such wrong subwords in the input, and measurements of $|\psi_{target}\rangle$ are independent, then the probability of acceptance of the wrong input is less than $\epsilon^m<\epsilon$. \end{proof}

\section{QFA implementation }\label{implementation}

In this section, we present a circuit  implementation of the automaton for the language $Palindrome_s$ for any $s$ using qiskit on a quantum simulator.  Several works are related with QFA implementation on noisy devices using 3 qubits such as  \cite{bsocy2021,sy2021}.  

Full program code can be found by the link \url{https://gitlab.com/aliya.khadi/palindromep-implementation/-/blob/main/Palindrome_qfa.ipynb}.

In the algorithm, we use multi-qubit controlled rotation. The circuit for this operation is decomposed into a circuit consisting of XGates (generalized  multi-qubit controlled not gate) and Ry gates. 
As we set a probability of error to $\epsilon$, the required number of qubits depends on $\epsilon$ and $s$. Let us consider how algorithm works when the number of qubits  is 4.

For example, the input word "$1001\#1111\#0110$" is accepted by the QFA with probability 1. "$11110111\#00011000\#00011101$" is rejected in 97 cases out of 100 ones.

\begin{longtable}{| c | c | c | }
\hline
\(w\) & \(result\) & \(error\) 

\\\hline
\endfirsthead
\hline
\endfoot

"$1001\#0000\#0101$"& is rejected & 0\\
 "$1001\#1111\#0110$"& is accepted & 0\\
"$110011\#000000\#000101$" &is rejected&0.04\\
"$11110011\#00011000\#00011101$"& is rejected & 0\\
"$11110111\#00011000\#00011101$"& is rejected & 0.03\\
"$11100111\#00000000\#10011001$"& is accepted & 0\\
"$0100110010\#0000100001\#1001001001\#1001001001$"& is rejected & 0.01\\
"$0100110010\#0000110000\#1001001001\#1101001011$"& is accepted & 0\\
\end{longtable}

\section{Conclusion}\label{conclusion}

The fingerprinting technique can be applied for a QFA solving a  promise problem $Palindrome_s$. This result shows the efficiency of quantum approach for solving the class of Equality-related problems. The algorithm is implemented by using qiskit tools for any values of $p$ and $\epsilon$.  

\paragraph*{Acknowledgements.}
We thank Kamil Khadiev and Abuzer Yakary{\i}lmaz for fruitful discussions and
help. 
The reported study was funded by the government assignment for FRC Kazan Scientific Center of RAS. This work was done during QCourse-511 held by the QWorld  Association.  A part of the study was funded by the subsidy allocated to
Kazan Federal University for the state assignment in the sphere of
scientific activities, project No. 0671-2020-0065.

\appendix \label{appen}
\section{Numerical Results}\label{experiment}

Parameters computing were performed on the following computer:
\begin{description}
\item [CPU:] Intel Core i7-5930K with 6 cores (12 threads) and 3.50 GHz frequency
\item [RAM:] 64 Gb
\item [GPU:] NVIDIA GPU GeForce GTX 1080 (GP104-A) with 8 Gb memory
\item [OS:] Ubuntu 18.04.1 LTS
\item [Octave version:] 4.2.2
\item [GPU drivers:] nvidia-drivers-415.23
\item [CUBLAS version:] 9.1
\end{description}

Results of numerical experiments are summarized in the following table.

The algorithm run for given parameters \(p\) and \(d\), and \(\epsilon\) was set
to \(\epsilon = 0.2\).

\begin{longtable}{| c | c | p{10cm} | c | }

\caption{Results of numerical experiments. Columns \(p\) and \(d\) contain input parameters. Column \(S\) contains \(d\) numbers --- the sequence of parameters that was found.  \(\epsilon=0.2\) 
}\\\hline
\(p\) & \(d\) & \(S\) 

\\\hline
\endfirsthead

\caption{(continued)}\\\hline
\(p\) & \(d\) & \(S\) 

\\\hline
\endhead

\hline
\endfoot

11 & 45 & [1, 4, 2, 5, 3, 1, 3, 5, 2, 4, 1, 3, 5, 2, 4, 5, 2, 1, 3, 4, 1, 4, 2, 3, 5, 1, 3, 5, 4, 2, 2, 5, 1, 3, 4, 2, 3, 4, 1, 5, 2, 3, 4, 5, 1]\\
17 & 51 & [2, 8, 3, 4, 7, 1, 5, 6, 2, 8, 5, 1, 6, 4, 3, 7, 1, 4, 6, 8, 3, 2, 5, 7, 1, 4, 6, 8, 3, 2, 5, 7, 1, 4, 6, 8, 3, 2, 5, 7, 1, 4, 7, 2, 5, 8, 6, 3, 1, 4, 6]\\
23 & 56 & [2, 5, 6, 9, 10, 1, 3, 7, 11, 4, 8, 1, 5, 8, 11, 2, 9, 4, 6, 7, 3, 10, 5, 2, 6, 9, 10, 1, 3, 7, 11, 4, 8, 1, 5, 8, 11, 2, 9, 4, 6, 7, 3, 10, 7, 11, 10, 8, 9, 6, 5, 4, 3, 1, 2, 2]\\
29 & 59 & [1, 8, 5, 11, 14, 2, 12, 4, 9, 7, 6, 10, 3, 13, 1, 8, 5, 11, 14, 2, 12, 4, 9, 7, 6, 10, 13, 3, 13, 2, 10, 7, 1, 4, 5, 8, 11, 14, 12, 9, 3, 6, 1, 11, 3, 5, 9, 7, 13, 14, 12, 10, 8, 6, 4, 2, 4, 3, 9]\\
31 & 60 & [1, 13, 8, 3, 6, 11, 4, 15, 9, 10, 2, 14, 5, 7, 12, 2, 5, 15, 6, 12, 9, 8, 1, 13, 11, 4, 3, 10, 14, 7, 7, 2, 6, 10, 11, 15, 3, 12, 1, 8, 14, 5, 4, 9, 13, 1, 12, 3, 5, 10, 8, 14, 6, 15, 4, 7, 13, 2, 9, 11]\\
37 & 63 & [1, 14, 6, 10, 2, 18, 11, 15, 7, 3, 5, 9, 13, 17, 16, 12, 8, 4, 1, 8, 11, 6, 15, 13, 10, 17, 3, 4, 18, 12, 5, 2, 9, 16, 7, 14, 8, 1, 11, 6, 15, 13, 10, 17, 3, 4, 18, 12, 5, 2, 9, 16, 14, 7, 12, 15, 16, 2, 5, 8, 9, 18, 1]\\
41 & 64 & [3, 14, 2, 20, 8, 15, 9, 4, 10, 16, 19, 13, 7, 1, 5, 11, 17, 18, 6, 12, 1, 9, 6, 19, 17, 4, 14, 12, 11, 7, 16, 2, 20, 3, 15, 8, 10, 13, 5, 18, 6, 13, 4, 1, 16, 11, 18, 8, 20, 9, 3, 15, 14, 2, 10, 19, 7, 5, 12, 17, 3, 14, 2, 20]\\
47 & 66 & [2, 11, 8, 7, 23, 20, 14, 5, 17, 1, 21, 4, 18, 15, 10, 12, 13, 9, 16, 6, 19, 3, 22, 3, 8, 15, 19, 2, 14, 9, 20, 21, 4, 10, 16, 22, 13, 7, 1, 5, 11, 17, 23, 18, 6, 12, 7, 15, 19, 1, 10, 23, 2, 6, 11, 20, 3, 14, 16, 18, 12, 5, 22, 8, 9, 21]\\
53 & 68 & [1, 8, 13, 17, 11, 15, 20, 6, 24, 10, 22, 3, 26, 4, 19, 12, 18, 5, 25, 2, 21, 9, 14, 16, 7, 23, 3, 7, 15, 13, 24, 1, 19, 11, 9, 17, 5, 21, 23, 25, 26, 22, 20, 18, 16, 14, 12, 10, 8, 6, 2, 4, 1, 8, 13, 17, 11, 15, 20, 6, 24, 10, 22, 3, 26, 4, 19, 12]\\
59 & 69 & [17, 15, 18, 4, 10, 12, 9, 20, 23, 7, 28, 1, 25, 26, 2, 6, 14, 22, 29, 21, 13, 5, 3, 11, 19, 27, 24, 16, 8, 8, 14, 5, 12, 29, 23, 27, 1, 10, 21, 25, 3, 16, 19, 6, 18, 17, 7, 28, 4, 20, 15, 9, 26, 2, 22, 13, 11, 24, 1, 9, 22, 26, 6, 19, 29, 12, 2, 16, 5]\\
61 & 70 & [1, 11, 16, 19, 7, 13, 28, 25, 4, 10, 30, 22, 2, 27, 5, 24, 8, 21, 18, 14, 15, 17, 12, 20, 9, 23, 6, 26, 3, 29, 4, 17, 3, 15, 28, 9, 10, 22, 16, 21, 2, 27, 8, 14, 20, 26, 29, 23, 11, 5, 1, 7, 13, 19, 25, 30, 24, 18, 12, 6, 11, 1, 7, 26, 16, 21, 3, 30, 17, 12]\\
67 & 71 & [11, 2, 19, 14, 18, 26, 10, 6, 23, 27, 31, 3, 32, 15, 7, 28, 24, 22, 1, 20, 5, 16, 30, 9, 12, 33, 13, 8, 29, 17, 4, 21, 25, 22, 13, 8, 20, 3, 18, 24, 1, 15, 29, 6, 33, 17, 31, 10, 27, 4, 26, 11, 19, 12, 25, 5, 32, 2, 28, 9, 21, 16, 14, 23, 30, 7, 15, 18, 25, 29, 1]\\
71 & 72 & [4, 34, 28, 27, 25, 11, 13, 5, 35, 19, 20, 3, 29, 12, 21, 26, 18, 10, 2, 6, 14, 22, 30, 33, 17, 9, 1, 7, 15, 23, 31, 32, 24, 8, 16, 12, 31, 13, 10, 4, 33, 32, 15, 34, 14, 11, 9, 16, 35, 8, 7, 17, 30, 6, 18, 29, 5, 19, 28, 20, 27, 3, 21, 26, 2, 22, 25, 1, 23, 24, 1, 21]\\
83 & 74 & [17, 6, 8, 38, 35, 1, 20, 22, 33, 4, 36, 24, 15, 10, 29, 31, 40, 13, 3, 27, 26, 34, 19, 11, 41, 12, 18, 5, 25, 28, 2, 32, 21, 9, 39, 14, 16, 37, 7, 30, 23, 1, 11, 29, 34, 8, 14, 31, 24, 36, 27, 6, 4, 39, 26, 9, 19, 21, 16, 41, 37, 32, 22, 17, 12, 7, 2, 3, 13, 18, 23, 28, 33, 38]\\
97 & 76 & [9, 5, 21, 24, 22, 41, 6, 37, 7, 40, 8, 25, 23, 10, 26, 38, 39, 27, 42, 43, 11, 44, 28, 4, 12, 20, 36, 45, 29, 13, 3, 19, 35, 46, 30, 14, 2, 18, 34, 47, 31, 15, 1, 17, 33, 48, 32, 16, 22, 12, 35, 7, 4, 37, 17, 36, 9, 6, 47, 20, 34, 32, 21, 19, 45, 48, 23, 8, 46, 18, 33, 5, 10, 38, 31, 24]\\

\\\hline
\end{longtable}


\begin{thebibliography}{99}

\bibitem{af98}
A.~Ambainis and R.~Freivalds, ``{\em 1-way quantum finite automata: strengths,
  weaknesses and generalizations},'' {\em FOCS'98}{\nolinebreak\hspace{0.1em}},
    332--341, IEEE (1998).

\bibitem{ABGKMT06}
A.~Ambainis, M.~Beaudry, M.~Golovkins, A.~\c{K}ikusts, M.~Mercer, and
  D.~Th{\'e}rien, ``{\em Algebraic results on quantum automata},'' {\em Theory
  of Computing Systems}~{\bf 39}(1),  165--188 (2006).

\bibitem{AKV01}
A.~Ambainis, A.~\c{K}ikusts, and M.~Valdats, ``{\em On the class of languages
  recognizable by 1-way quantum finite automata},'' {\em STACS 2001: 18th
  Annual Symposium on Theoretical Aspects of Computer
  Science}{\nolinebreak\hspace{0.1em}},   75--86 (2001).

\bibitem{ABFK99}
A.~Ambainis, R.~Bonner, R.~Freivalds, and A.~\c{K}ikusts, ``{\em Probabilities
  to accept languages by quantum finite automata},'' {\em Computing and
  Combinatorics}{\nolinebreak\hspace{0.1em}},  {\em Lecture Notes in Computer
  Science} {\bf 1627},  174--183, Springer Berlin / Heidelberg (1999).

\bibitem{AY15}
A.~Ambainis and A.~Yakary{\i}lmaz, ``{\em Automata and quantum computing},''
  Tech. Rep. 1507.01988, arXiv (2015).

\bibitem{YS10A}
A.~Yakary{\i}lmaz and A.~C.~C. Say, ``{\em Languages recognized by
  nondeterministic quantum finite automata},'' {\em Quantum Information and
  Computation}~{\bf 10}(9\&10),  747--770 (2010).

\bibitem{sy2014}
A.~C. Say and A.~Yakary{\i}lmaz, ``{\em Quantum finite automata: A modern
  introduction},'' {\em Computing with New
  Resources}{\nolinebreak\hspace{0.1em}},   208--222, Springer (2014).

\bibitem{gy2018}
A.~Gainutdinova and A.~Yakary{\i}lmaz, ``{\em Unary probabilistic and quantum
  automata on promise problems},'' {\em Quantum Information Processing}~{\bf
  17}(2),  28 (2018).

\bibitem{AK03}
A.~Ambainis and A.~\c{K}ikusts, ``{\em Exact results for accepting
  probabilities of quantum automata},'' {\em Theoretical Computer Science}~{\bf
  295}(1-3),  3--25 (2003).

\bibitem{ZQGLM13}
S.~Zheng, D.~Qiu, J.~Gruska, L.~Li, and P.~Mateus, ``{\em State succinctness of
  two-way finite automata with quantum and classical states},'' {\em
  Theoretical Computer Science}~{\bf 499},  98--112 (2013).

\bibitem{AY12}
A.~Ambainis and A.~Yakary{\i}lmaz, ``{\em Superiority of exact quantum automata
  for promise problems},'' {\em Information Processing Letters}~{\bf 112}(7),
  289--291 (2012).

\bibitem{av2009}
F.~Ablayev and A.~Vasiliev, ``{\em On quantum realisation of boolean functions
  by the fingerprinting technique},'' {\em Discrete Mathematics and
  Applications}~{\bf 19}(6),  555--572 (2009).

\bibitem{av2008}
F.~Ablayev and A.~Vasiliev, ``{\em On the computation of boolean functions by
  quantum branching programs via fingerprinting},'' {\em Electronic Colloquium
  on Computational Complexity (ECCC)}{\nolinebreak\hspace{0.1em}},   {\bf
  15}(059) (2008).

\bibitem{av2013}
F.~M. Ablayev and A.~Vasiliev, ``{\em Algorithms for quantum branching programs
  based on fingerprinting.},'' {\em Int. J. Software and Informatics}~{\bf
  7}(4),  485--500 (2013).

\bibitem{kk2017}
K.~Khadiev and A.~Khadieva, ``{\em Reordering method and hierarchies for
  quantum and classical ordered binary decision diagrams},'' {\em CSR
  2017}{\nolinebreak\hspace{0.1em}},  {\em LNCS} {\bf 10304},  162--175,
  Springer (2017).

\bibitem{l2009}
F.~Le~Gall, ``{\em Exponential separation of quantum and classical online space
  complexity},'' {\em Theory of Computing Systems}~{\bf 45}(2),  188--202
  (2009).

\bibitem{l2006}
F.~Le~Gall, ``{\em Exponential separation of quantum and classical online space
  complexity},'' {\em SPAA '06},  67--73, ACM (2006).

\bibitem{kk2022}
K.~Khadiev and A.~Khadieva, ``{\em Quantum and classical log-bounded automata
  for the online disjointness problem},'' {\em Mathematics}~{\bf 10}(1) (2022).

\bibitem{aaksv2022}
F.~Ablayev, M.~Ablayev, K.~Khadiev, N.~Salihova, and A.~Vasiliev, ``{\em
  Quantum algorithms for string processing},'' {\em Mesh Methods for
  Boundary-Value Problems and Applications}{\nolinebreak\hspace{0.1em}},  {\em
  Lecture Notes in Computational Science and Engineering} {\bf 141} (2022).

\bibitem{kk2019disj}
K.~Khadiev and A.~Khadieva, ``{\em Quantum online streaming algorithms with
  logarithmic memory},'' {\em International Journal of Theoretical
  Physics}~{\bf 60},  608--616 (2021).

\bibitem{aavz2016}
F.~Ablayev, M.~Ablayev, A.~Vasiliev, and M.~Ziatdinov, ``{\em Quantum
  fingerprinting and quantum hashing. computational and cryptographical
  aspects},'' {\em Baltic Journal of Modern Computing}~{\bf 4}(4),  860 (2016).

\bibitem{aakv2018}
F.~Ablayev, M.~Ablayev, K.~Khadiev, and A.~Vasiliev, ``{\em Classical and
  quantum computations with restricted memory},'' {\em LNCS}~{\bf 11011},
  129--155 (2018).

\bibitem{an2008}
A.~Ambainis and N.~Nahimovs, ``{\em Improved constructions of quantum
  automata.},'' {\em TQC}{\nolinebreak\hspace{0.1em}},  {\em LNCS} {\bf 5106},
  47--56, Springer (2008).

\bibitem{an2009}
A.~Ambainis and N.~Nahimovs, ``{\em Improved constructions of quantum
  automata},'' {\em Theoretical Computer Science}~{\bf 410}(20),  1916--1922
  (2009).

\bibitem{lvz2017}
A.~Vasiliev, M.~Latypov, and M.~Ziatdinov, ``{\em Minimizing collisions for
  quantum hashing},'' {\em Journal of Engineering and Applied Sciences}~{\bf
  12}(4),  877--880 (2017).

\bibitem{Weg00}
I.~Wegener,  {\em Branching Programs and Binary Decision Diagrams: Theory and
  Applications}{\nolinebreak\hspace{0.1em}}, SIAM (2000).

\bibitem{agk01}
F.~Ablayev, A.~Gainutdinova, and M.~Karpinski, ``{\em On computational power of
  quantum branching programs},'' {\em FCT}{\nolinebreak\hspace{0.1em}},  {\em
  LNCS} {\bf 2138},  59--70, Springer (2001).

\bibitem{agkmp2005}
F.~Ablayev, A.~Gainutdinova, M.~Karpinski, C.~Moore, and C.~Pollett, ``{\em On
  the computational power of probabilistic and quantum branching program},''
  {\em Information and Computation}~{\bf 203}(2),  145--162 (2005).

\bibitem{ag05}
F.~Ablayev and A.~Gainutdinova, ``{\em Complexity of quantum uniform and
  nonuniform automata},'' {\em Developments in Language
  Theory}{\nolinebreak\hspace{0.1em}},  {\em LNCS} {\bf 3572},  78--87,
  Springer (2005).

\bibitem{a2006}
F.~Ablayev, ``{\em On computational power of classical and quantum branching
  programs},'' {\em Proc.SPIE}~{\bf 6264},  6264 -- 6264 -- 10 (2006).

\bibitem{fmp02}
F.~Ablayev, C.~Moore, and C.~Pollett, ``{\em Quantum and stochastic branching
  programs of bounded width},'' {\em Automata, Languages and
  Programming}{\nolinebreak\hspace{0.1em}},   343--354, Springer Berlin
  Heidelberg, Berlin, Heidelberg (2002).

\bibitem{ikpy2018}
R.~Ibrahimov, K.~Khadiev, K.~Pr\={u}sis, and A.~Yakary{\i}lmaz, ``{\em
  Error-free affine, unitary, and probabilistic \mbox{OBDD}s},'' {\em Lecture
  Notes in Computer Science}~{\bf 10952 LNCS},  175--187 (2018).

\bibitem{ikpy2021}
R.~Ibrahimov, K.~Khadiev, K.~Pr{\=u}sis, and A.~Yakary{\i}lmaz, ``{\em
  Error-free affine, unitary, and probabilistic obdds},'' {\em International
  Journal of Foundations of Computer Science}~{\bf 32}(7),  849--860 (2021).

\bibitem{aakk2018}
F.~Ablayev, A.~Ambainis, K.~Khadiev, and A.~Khadieva, ``{\em Lower bounds and
  hierarchies for quantum memoryless communication protocols and quantum
  ordered binary decision diagrams with repeated test},'' {\em In SOFSEM 2018,
  LNCS}~{\bf 10706},  197--211 (2018).

\bibitem{agky14}
F.~Ablayev, A.~Gainutdinova, K.~Khadiev, and A.~Yakary{\i}lmaz, ``{\em Very
  narrow quantum \mbox{OBDD}s and width hierarchies for classical
  \mbox{OBDD}s},'' {\em DCFS}{\nolinebreak\hspace{0.1em}},  {\em LNCS} {\bf
  8614},  53--64, Springer (2014).

\bibitem{agky16}
F.~Ablayev, A.~Gainutdinova, K.~Khadiev, and A.~Yakary{\i}lmaz, ``{\em Very
  narrow quantum \mbox{OBDD}s and width hierarchies for classical
  \mbox{OBDD}s},'' {\em Lobachevskii Journal of Mathematics}~{\bf 37}(6),
  670--682 (2016).

\bibitem{wx2008}
A.~Wigderson and D.~Xiao, ``{\em Derandomizing the ahlswede-winter
  matrix-valued chernoff bound using pessimistic estimators, and
  applications},'' {\em Theory of Computing}~{\bf 4},  53--76 (2008).

\bibitem{MC00}
C.~Moore and J.~P. Crutchfield, ``{\em Quantum automata and quantum
  grammars},'' {\em Theoretical Computer Science}~{\bf 237}(1-2),  275--306
  (2000).

\bibitem{mrr95}
R.~Motwani and P.~Raghavan,  {\em Randomized
  algorithms}{\nolinebreak\hspace{0.1em}}, Cambridge university press (1995).

\bibitem{a1990}
M.~Ajtai, H.~Iwaniec, J.~Koml{\'o}s, J.~Pintz, and E.~Szemer{\'e}di, ``{\em
  Construction of a thin set with small fourier coefficients},'' {\em Bulletin
  of the London Mathematical Society}~{\bf 22}(6),  583--590 (1990).

\bibitem{akv2008}
F.~Ablayev, A.~Khasianov, and A.~Vasiliev, ``{\em On complexity of quantum
  branching programs computing equality-like boolean functions},'' {\em ECCC}
  (2010).

\bibitem{bsocy2021}
U.~Birkan, {\"O}.~Salehi, V.~Olejar, C.~Nurlu, and A.~Yakary{\i}lmaz, ``{\em
  Implementing quantum finite automata algorithms on noisy devices},'' {\em
  International Conference on Computational
  Science}{\nolinebreak\hspace{0.1em}},   3--16, Springer (2021).

\bibitem{sy2021}
{\"O}.~Salehi and A.~Yakary{\i}lmaz, ``{\em State-efficient qfa algorithm for
  quantum computers},'' {\em arXiv preprint arXiv:2107.02262}  (2021).

\end{thebibliography}
\end{document}